\documentclass[aps,prl,letterpaper,superscriptaddress,twocolumn]{revtex4}

\usepackage{color}
\usepackage{amssymb,bm}
\usepackage{amsthm}
\usepackage{amsmath}
\usepackage{graphicx}
\usepackage{times}
\usepackage{algorithm}
\floatname{algorithm}{Protocol}

\usepackage[T1]{fontenc}
\usepackage[utf8]{inputenc}
\usepackage{graphicx}
\usepackage[colorlinks]{hyperref}

\def\tr{\mbox{Tr}}
\def\be{\begin{equation}}
\def\ee{\end{equation}}

\newcommand{\bra}[1]{\left\langle{#1}\right\vert}
\newcommand{\ket}[1]{\left\vert{#1}\right\rangle}

\theoremstyle{definition}

\theoremstyle{theorem}
\newtheorem{theorem}{Theorem}
\newtheorem{lemma}[theorem]{Lemma}

\theoremstyle{definition}

\usepackage{color}

\def\>{\rangle} 
\def\<{\langle}  

\begin{document}

\title{Causal Limit on Quantum Communication}

\author{Robert Pisarczyk} 
\email{robert.pisarczyk@maths.ox.ac.uk}

\address{Mathematical Institute, University of Oxford, Woodstock Road, Oxford OX2 6GG, U.K.} 
\address{Centre for Quantum Technologies, National University of Singapore, 3 Science Drive 2, 117543, Singapore}

\author{Zhikuan Zhao}
\email{zhikuan.zhao@inf.ethz.ch}
\affiliation{Department of Computer Science, ETH Zürich, Universitätstrasse 6, 8092 Zürich}
\affiliation{Centre for Quantum Technologies, National University of Singapore, 3 Science Drive 2, 117543, Singapore}
\affiliation{Singapore University of Technology and Design, 8 Somapah Road, 487372, Singapore}

\author{Yingkai Ouyang}
\email{y.ouyang@sheffield.ac.uk}
\affiliation{University of Sheffield, Department of Physics and Astronomy, 226 Hounsfield Rd, Sheffield S3 7RH, U.K.}
\affiliation{Centre for Quantum Technologies, National University of Singapore, 3 Science Drive 2, 117543, Singapore}
\affiliation{Singapore University of Technology and Design, 8 Somapah Road, 487372, Singapore}

\author{Vlatko Vedral}
\email{vlatko.vedral@gmail.com}
\affiliation{Centre for Quantum Technologies, National University of Singapore, 3 Science Drive 2, 117543, Singapore}
\affiliation{Clarendon Laboratory, Department of Physics, University of Oxford, Parks Road, Oxford OX1 3PU, U.K.}
\affiliation{Department of Physics, National University of Singapore, 2 Science Drive 3, 117542, Singapore}

\author{Joseph F. Fitzsimons}
\email{joe.fitzsimons@nus.edu.sg}
\affiliation{Centre for Quantum Technologies, National University of Singapore, 3 Science Drive 2, 117543, Singapore}
\affiliation{Singapore University of Technology and Design, 8 Somapah Road, 487372, Singapore}
\affiliation{Horizon Quantum Computing, 79 Ayer Rajah Crescent, Singapore 139955}

\begin{abstract}
The capacity of a channel is known to be equivalent to the highest rate at which it can generate entanglement. Analogous to entanglement, the notion of a causality measure characterises the temporal aspect of quantum correlations. Despite holding an equally fundamental role in physics, temporal quantum correlations have yet to find their operational significance in quantum communication. Here we uncover a connection between quantum causality and channel capacity. We show the amount of temporal correlations between two ends of the noisy quantum channel, as quantified by a causality measure, implies a general upper bound on its channel capacity. The expression of this new bound is simpler to evaluate than most previously known bounds. We demonstrate the utility of this bound by applying it to a class of shifted depolarizing channels, which results in improvement over previously calculated bounds for this class of channels.
\end{abstract}

\date{\today}
\maketitle

\textit{Introduction.--}Determining the rate at which information can be reliably transmitted over a given channel is one of the central tasks of information theory. In a classical setting, Shannon \cite{shannon1948mathematical} proved that the capacity of discrete memoryless channels are governed by a simple expression. In a quantum setting, however, such a characterisation of a channels' ability to transmit information has proved far more elusive. In determining the capacity of a quantum channel, $\mathcal{N}$, we have to consider the possibility that in order to achieve the maximal capacity per use of the channel it may be necessary to encode information in states which are entangled across channels. Thus, to determine the actual capacity of a quantum channel, one needs to take the supremum of this quantity over tensor products of an arbitrary number of copies of the channel. In the context of quantum communication, a significant amount of progress has been made on achievable rates for the transmission of quantum information over noisy channels  \cite{holevo2001evaluating,lloyd1997capacity, shor2002quantum, devetak2005private}. However existing formulae for quantum capacities often involve implicit optimisation problems. In the absence of formulae for the exact capacities, one is forced to rely on bounds for the quantum capacity that are tractable to evaluate \cite{takeoka2014squashed, muller2016positivity, wang2016semidefinite, sutter2015approximate,pirandola2017fundamental, pirandola2019fundamental, wang2019semidefinite, tomamichel2017strong, berta2018amortization}. The reader is referred to \cite{wilde2013quantum, hayashi2006quantum} for a review of related results.

The quantum capacity is also known to be equivalent to the highest rate at which the channel can be used to generate quantum entanglement, the essential nonclassical signature in composite quantum systems \cite{wilde2013quantum}. While the conceptual link between channel capacity and spatial quantum correlations has become increasingly clear, the operational role of temporal correlations in quantum communication remains to be clearly depicted. Powerful existing frameworks such as the process matrices \cite{brukner2014quantum, oreshkov2012quantum} have enabled novel results in a setting where the causal order in a communication task is indefinite \cite{chiribella2018indefinite, salek2018quantum, jia2019causal} while the framework of quantum causal models has been employed to study cause-effect and temporal relations between quantum systems \cite{allen2017quantum, barrett2019quantum}. Here we work in the conventional setting of one-way quantum communication and integrate causal considerations into the traditional framework of quantum Shannon theory. Specifically, we view a quantum communication process through a noisy channel as a generalised quantum state that is extended across time. Taking this viewpoint intuitively connects the channel's quantum capacity with its ability to preserve causal correlations between the input and output.

In this letter, we present novel general upper bounds on the quantum capacities of quantum channels that do not require optimisation and are based on causality considerations derived using a pseudo-density matrix (PDM) formalism introduced in \cite{fitzsimons2015quantum}, with the bound also expressible in terms of the Choi matrix of a channel \cite{choi1975completely}. A PDM is a generalization of the standard density matrix which seeks to capture both spatial and temporal correlations. In quantum mechanics, a density matrix is a probability distribution over pure quantum states but it can alternatively be viewed as a representation of the expectation values for each possible Pauli measurement on the system. For a system composed of multiple spatially separated subsystems, each Pauli operator can be expanded as a tensor product of single-qubit Pauli operators, with one acting on each subsystem. PDMs build on this second view of the standard density matrix, extending the notion of the density matrix into the time domain. The resulting pseudo-density matrix is defined as 
\[
R = \frac{1}{2^n} \sum_{i_1 = 0}^{3} ... \sum_{i_n = 0}^{3} \left \langle\{\sigma_{i_j}\}_{j=1}^n \right \rangle \bigotimes_{j=1}^n \sigma_{i_j},
\]
where $\left \langle\{\sigma_{i_j}\}_{j=1}^n \right \rangle$ is the expectation value for the product of a set of Pauli measurements. Unlike in the standard density matrix, we do not require the measurements act only on distinct spatially separated subsystems.
Rather each measurement can be associated with an instant in time and a particular subsystem, and is taken to project the state of the system onto the eigenspace of the measured observable corresponding to the measurement outcome. We also note that although the PDM is introduced with respect to the set of qubits, it can describe a quantum system of any dimensionality. One needs to embed such a system into a state of qubits and restrict its evolution to the appropriate subspace.

\textit{Causality monotone.--}The generalization of states to systems extended across multiple points in time has the result that, unlike density matrices, PDMs can have negative eigenvalues. As the PDM is equivalent to the standard density matrix when the measurements are restricted to a single moment in time, the existence of negative eigenvalues in the PDM acts as a witness to temporal correlations in the measurement events. In order to quantify the causal component of such correlations, the notion of a causality monotone was introduced in \cite{fitzsimons2015quantum}. We now introduce a function based on the logarithm of the trace norm of the PDM, $F(R) = \log_2 \|R\|_1$, which is similar to causality monotones, but sacrifices convexity in favour of additivity when applied to tensor products. This is similar to logarithmic negativity \cite{plenio2005logarithmic} in the context of spatial correlations. Analogous to entanglement measures \cite{vedral1998entanglement},  $F(R)$ satisfies the following important properties:
\begin{enumerate}
\item $F(R) \geq 0$, with $F(R)=0$ if $R$ is positive semi-definite, and $F(R_2) = 1$ for $R_2$ obtained from two consecutive measurements on a single qubit closed system,
\item $F(R)$ is invariant under a local change of basis,
\item $F(R)$ is non-increasing under local operations,
\item $F(\sum_i p_i R_i) \leq \max_i F(R_i)$, for any probability distribution $\{p_i\}$, and
\item $F(R\otimes S) = F(R) + F(S)$.
\end{enumerate}
Properties 1-3 follow directly from the corresponding properties of the causality monotone $f_\text{tr}(R) = \|R\|_1 - 1$ proved in \cite{fitzsimons2015quantum}, since $F(R) = \log_2 (f_\text{tr}(R)+1)$, and from the monotonicity of the logarithm function. Property 4 also follows from the monotonicity of the logarithm function, since this implies 
$F(\sum_i p_i R_i) 
\leq 
\max_i F(R_i\sum_j p_j) 
= 
\max_i F(R_i)$. 
To prove property 5, we note that $\log_2 \|R\otimes S\|_1 = \log_2 \|R\|_1 \|S\|_1 = \log_2 \|R\|_1 + \log_2 \|S\|_1 $, and hence $F(R\otimes S) = F(R) + F(S)$.

\textit{Causality bound on quantum channel capacity.--} Evolution of any quantum state can be identified with a corresponding PDM. Consider a qubit-to-qubit channel $\mathcal N_1$ acting on a single qubit described by an initial state $\rho$. For such a process $R_{\mathcal N_1}$, a PDM that involves a single use of the channel $\mathcal N_1$ and two measurements before and after $\mathcal N_1$, has been shown to be given by
\begin{equation}
R_{\mathcal N_1}=	(\mathcal{I}\otimes \mathcal N_1) ( \{\rho\otimes\frac{\mathrm{I}}{2},\textrm{SWAP}\}), \label{single}
\end{equation}
where $\textrm{SWAP}=\frac{1}{2}\sum_{i=0}^3\sigma_i\otimes\sigma_i$ and $\{A,B\}=AB+BA$ \cite{horsman2017can, zhao2018geometry}. Here we fix the input $\rho$ to be a maximally mixed state. Then, equation \eqref{single} can be easily generalised to describe any quantum channel $\mathcal{N}$ acting on a collection of $l$ qubits
\begin{equation}
R_{\mathcal N}=(\mathcal{I}\otimes \mathcal{N})\left(\frac{\textrm{SWAP}^{\otimes l}}{2^l}\right). \label{SWAP}
\end{equation}
It is worth noting that choosing $\rho$ to be maximally mixed, the causality measure $F(R_{\mathcal{N}})$ gains a simple interpretation as it reduces to the logarithmic negativity of the Choi state of $\mathcal{N}$, which measures the amount of entanglement preserved in an initially maximally entangled two-qubit system after a subsystem is sent through the channel. As such the well-studied entanglement measure, negativity has an equally valid role in the temporal domain, in that it quantifies a channel's ability to preserve causal correlations. The relevance of the Choi state for causal structures in quantum mechanics has also been found in previous work \cite{allen2017quantum, barrett2019quantum}. Here we take the novel step to directly link the properties of a Choi state with the  quantum capacity of the respective channel.

Operationally, the quantum capacity of a quantum channel ${\mathcal N}$ is the maximum rate in which quantum information can be transmitted across $n$ independent uses of the quantum channel ${\mathcal N}$ with vanishing error as the number of uses $n$  approaches infinity. Therefore in order to relate the causality measure to the quantum channel capacity,  we employ equation \eqref{SWAP} and use the causality measure $F(R_{\mathcal N})$ to construct an upper bound on the number of uses of a given channel $\mathcal{N}$ to approximate the ideal (identity) channel $\mathcal I ^{\otimes k}$. As in the canonical setting of \cite{lloyd1997capacity} we wish to approximate $k$ copies of the identity channel as it corresponds precisely to the asymptotically perfect transmission of $k$ copies of a state. Since we consider only one-way communication in the memoryless setting, the most general procedure for combining resource channels together to approximate the ideal channel is to consider $n$ parallel uses of the channel preceded by some encoding and followed by some decoding procedure, as shown below in FIG.~\ref{fig:Enc}. We do not consider memory effects in our work. However, it would be interesting to extend our results to the capacities of quantum channels with memory in a future study \cite{kretschmann2005quantum}.
\begin{figure}[H]
	\centering
	\includegraphics[width=1.0\linewidth]{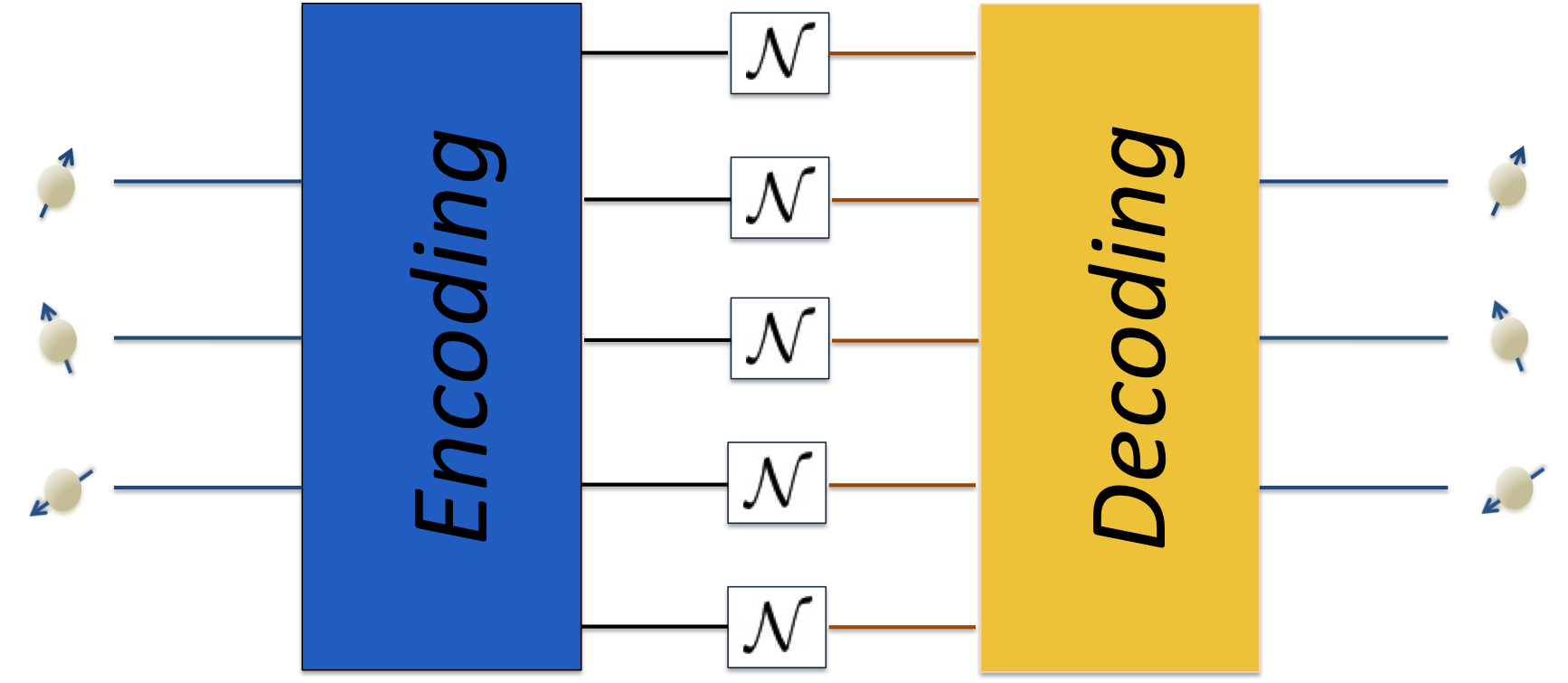}
	\caption{A quantum state of a collection of $k$ qubits is encoded into a larger Hilbert space. The encoded quantum information is sent through $n$ parallel copies of the resource channel $\mathcal N$ after which it decoded. In general, the dimensions of the input and the output of channel $\mathcal{N}$ need not be the same. As encoding and decoding are both physical processes, they are completely positive trace preserving maps.}
	\label{fig:Enc}
\end{figure}
We compare the causality measure across the collection of channels with the causality measure across the identity channel. As a result of property 3 of $F(R)$ and the fact that for quantum channel capacity consideration it suffices to consider isometric encodings \cite{barnum2000quantum}, the causality measure across the combined channels does not increase under encoding and decoding. We then exploit the additivity of causality measure to relate $k$ to the number of uses of the channel. In fact, the same properties of $F$ guarantee that even if we had allowed the encoding and decoding procedures to operate on entangled ancillary registers, the above relations would still hold, and hence the bounds we derive from this will also upper bound the entanglement-assisted capacities $\mathcal{N}$ \cite{bennett1999entanglement, bennett2002entanglement}. This leads to our main result that the quantum capacity $Q$ of channel $\mathcal N$ is upper bounded by $F(R_{\mathcal N})$,
\begin{equation}
Q(\mathcal N) \leq F(R_{\mathcal N}). \label{eq:result}
\end{equation}
The mathematical details for deriving this bound are presented in the section below.  Evaluating the causality measure $F(R_{\mathcal N})$ requires only finding the logarithm of the trace norm of a PDM and can be readily calculated for channels acting on relatively small Hilbert spaces. Importantly, computing this bound does not involve any optimisation. Furthermore, equation \eqref{eq:result} implies that any channel with $F(R_{\mathcal{N}}) = 0$ has quantum capacity equal to zero. This reflects the fact that such a channel exhibits correlations which could have been produced by measurements on distinct subsystems of a quantum state, and so the system is necessarily constrained by the no-signalling theorem. On the other hand, when $F(R_{\mathcal{N}})$ is strictly positive, the correlations between the two ends of the channel cannot be captured by bipartite density matrices, thus signifying information being passed forward in time. We emphasize that the bound has been derived for channels acting on the collection of qubits, nonetheless the result applies to channels with arbitrary input and output dimensions. For the method to apply to such cases, it suffices to embed the system into a  $2^k$ dimensional Hilbert space of qubits and
restrict the channel to act only on a subspace of this space.

\begin{figure}[h!]
	\centering
	\includegraphics[width=1.0\linewidth]{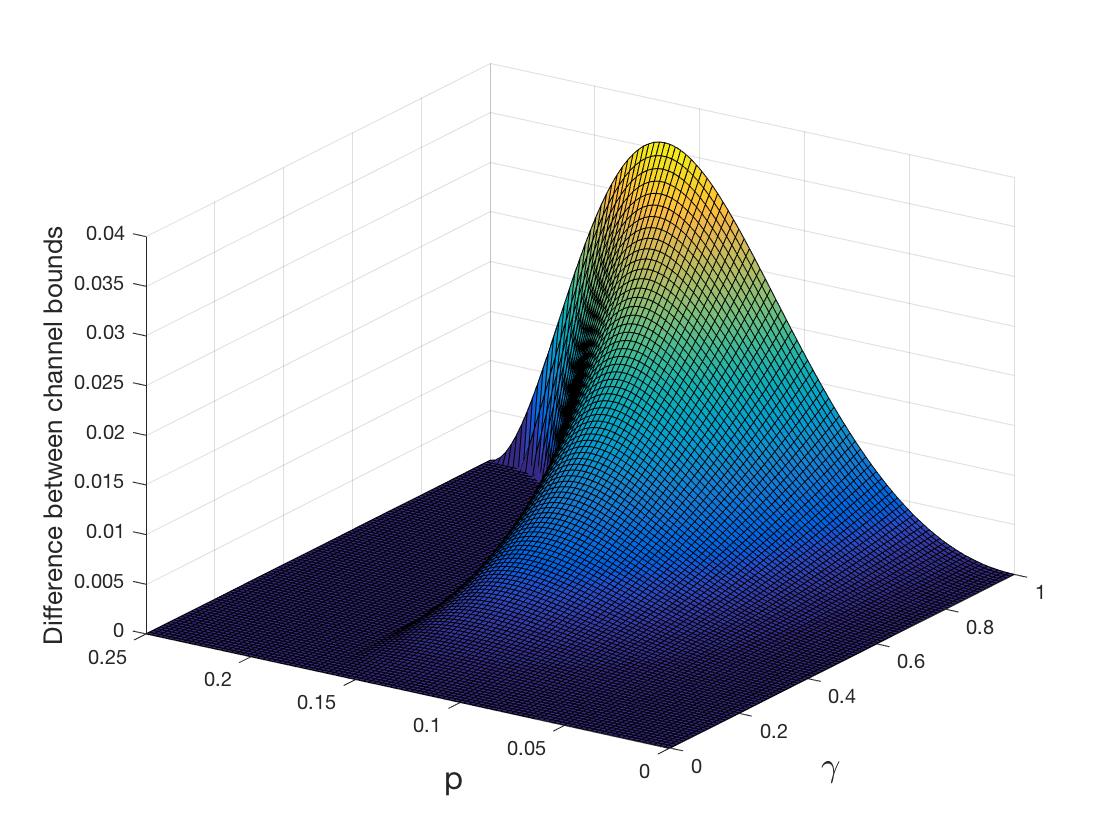}
	\caption{Difference between the HW and causality bounds on quantum channel capacity of a shifted depolarizing channel. Notice that the two bounds coincide when there is no shift (standard depolarizing channel) but the causality bound is tighter when the shift $\gamma$ increases.}
	\label{fig:HW}
\end{figure}
It is also interesting to note the apparent resemblance between the causality bound and the max-Rains information bound \cite{wang2019semidefinite} which is also expressible through properties of the Choi state. Indeed, in the Supplemental Material, we show that the max-Rains bound for a channel $\mathcal N$ is upper bounded by the causality bound for the conjugate channel ${\mathcal N}^*$. As a result, the max-Rains bound might often be a tighter bound. In contrast, the causality bound is not a semi-definite program, requires no optimisation and as such is analytically calculable \cite{wilde2018entanglement}. Furthermore, the max-Rains information provides a bound for the distillable entanglement of a channel, which is a related but distinct concept from the distillable entanglement of a state. The distillable entanglement of a state is known to be upper-bounded by logarithmic negativity while our bound relates logarithmic negativity to the distillable entanglement of channels.

\textit{Application of the bound.--}As a practical illustration of how the causality method works, we apply it to the class of shifted depolarizing channels. A shifted depolarizing channel generalises the well-studied quantum depolarizing channel \cite{king2003capacity,smith2008additive, leditzky2018useful}. It outputs either the state $\frac{{I}+\gamma Z}{2}$ shifted from the maximally mixed state with probability $4p$ or the input state. For a single qubit the channel can be defined by $\mathcal{N}_{\gamma} (\rho) = (1-4p)\rho + 4p\left(\frac{{I}+\gamma Z}{2} \right)$. The parameter $\gamma \in [0,1]$ parametrizes the shift, with vanishing $\gamma$ corresponding to a standard depolarizing channel. The PDM $R_{\mathcal{N}_{\gamma}}$ associated with the single qubit shifted depolarizing channel can be found using equation \eqref{SWAP} from which we obtain an analytic expression for the value of $F(R_{\mathcal{N}_{\gamma}})$, and hence an upper bound on the quantum capacity of the channel

\begin{align*}
\mathcal{Q}(\mathcal{N}_{\gamma}) &\leq F(R_{\mathcal{N}_{\gamma}}) \nonumber \\
&= \log_2 \bigg( 1-p + \frac{1}{2} \sqrt{1-8p+16 p^2+4\gamma^2 p^2} \nonumber\\
&\hspace{1.2cm}+\frac{1}{2}\left| 2p-\sqrt{1-8p+16p^2+4\gamma^2 p^2}\right|\bigg).
\end{align*}
We can compare this with a simple well-known bound on quantum capacities of Holevo and Werner (HW) which is general, and has a similar form to the causality bound, but requires optimisation \cite{holevo2001evaluating}. The causality bound is better or equal to the HW bound (see the Supplemental Material for a proof). As shown in FIG. \ref{fig:HW}, the shifted depolarising channel constitutes an example for which the causality bound is strictly tighter than the HW bound.
Furthermore, the bound $F(R_{\mathcal{N}_{\gamma}})$ also improves upon the best known bound from \cite{ouyang2014channel}. In fact, it is tighter for most values of shifts $\gamma$ as shown in FIG.~\ref{fig:Kai}. We should also note that the causality bound coincides with the max-Rains bound for the shifted depolarizing channel \cite{Wang}.

\begin{figure}[h!]
	\centering
	\includegraphics[width=1.0\linewidth]{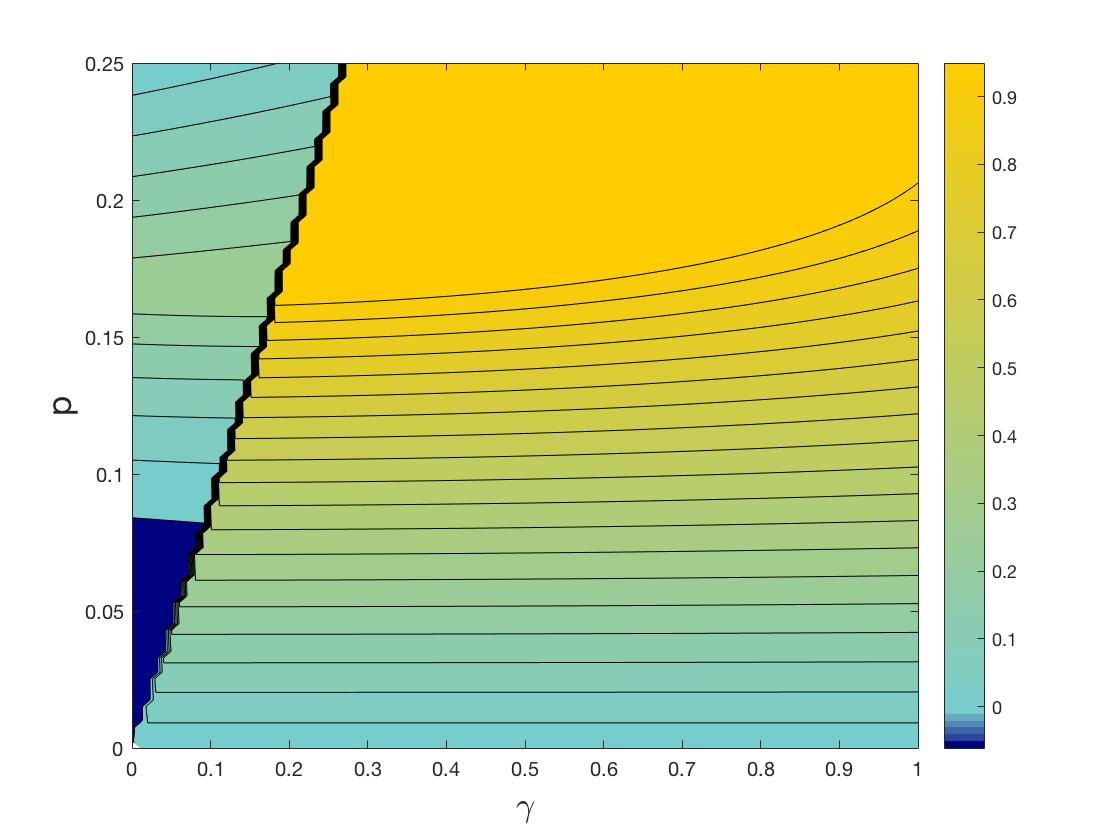}
	\caption{Difference between the previously known bound from \cite{ouyang2014channel} and the causality bound on quantum channel capacity of a shifted depolarizing channel. The causality bound is tighter for almost all values of $\gamma$ and $p$. Only in the region of small shift $\gamma$ and small probability $p$, which corresponds to the bottom left corner of the diagram, the causality bound is less tight.}
	\label{fig:Kai}
\end{figure}

\textit{Proof of the bound.--}In this section, we prove the bound in equation \eqref{eq:result}.
First, we construct the pseudo-density matrix corresponding to a channel obtained through using $n$ copies of the resource channel $\mathcal{N}$ preceded by the encoding channel $\mathcal{E}$ and followed the decoding channel $\mathcal{D}$. Let $\mathcal{M} = \mathcal{D} \circ \mathcal{N}^{\otimes n} \circ \mathcal{E}$. Note that 
\[R_\mathcal{M} = ({\mathcal I}^{\otimes k} \otimes \mathcal{M})(R_{\mathcal{I} ^{\otimes k}}).
\]
By the reverse triangle inequality,
\[
\| R_{\mathcal{M}}\| _1
=
\|R_{{\mathcal I ^{\otimes k}}}  +  R_{\mathcal{M}} -  R_{\mathcal I ^{\otimes k}}  \| _1
\ge
\|R_{{\mathcal I ^{\otimes k}}} \|_1 -  \|R_{\mathcal{M}} -  R_{\mathcal I ^{\otimes k}}  \|_1 .
\]
We can relate the trace distance of two pseudo-density matrices to the diamond norm in the following way:
\begin{align*}
\|R_{\mathcal{M}} - R_{{\mathcal I} ^{\otimes k}}\| _1
&= 
\|({\mathcal I} ^{\otimes k}\otimes 
(\mathcal{M} -  {\mathcal I}^{\otimes k}))(R_{\mathcal I ^{\otimes k}}) \|_1  \\
&\leq 
\|
\mathcal{M} -  {\mathcal I}^{\otimes k} \|_\diamond
\|R_{\mathcal I ^{\otimes k}}\|_1 ,
\end{align*}
where $\|\cdot\|_\diamond$ denotes the diamond norm \cite{kitaev2002classical}. Denoting the distance between $\mathcal M$ and $\mathcal I$ in the diamond norm by $\epsilon=  \|\mathcal{M} -  {\mathcal I}^{\otimes k} \|_{\diamond}$ and using the upper bound on $\|R_{\mathcal{M}} - R_{{\mathcal I} ^{\otimes k}}\|_1 $ as well as the positivity of
$\|R_{{\mathcal I ^{\otimes k}}}\|_1$,
we get
\[
\frac{ \| R_{\mathcal{M}}\|_1 }{\|R_{{\mathcal I ^{\otimes k}}} \|_1  } 
\ge
1 - \epsilon. 
\] 
Taking the logarithm on both sides of the above inequality, we find
\[
F(R_{\mathcal M}) - F(R_{\mathcal I ^{\otimes k}}) \ge \log_2(1-\epsilon) .
\] 
We can exploit the relation between the PDM and \textrm{SWAP} matrix, as well as the non-increasing property of the trace norm under the partial trace, to show that the causality measure does not increase under decoding and encoding. A detailed proof is presented in the Supplemental Material. This gives  $F(R_{\mathcal{M}}) \leq F(R_{\mathcal{N}}^{\otimes n})$.

Additivity of $F$ with respect to tensor products implies that $F(R_{\mathcal{N}}^{\otimes n}) = n F(R_{\mathcal N})$, and $F(R_{\mathcal I ^{\otimes k}})=k F(R_{\mathcal I})$.
Hence
\[
n F(R_{\mathcal N}) - k  F(R_{\mathcal I}) \ge \log_2(1-\epsilon) .
\]  
Finally, since $F(R_{\mathcal I})=l$, where $l$ is the number of qubits on which the channel acts, we have
\[
\frac{l k}{n} \le F(R_{\mathcal N})  - \frac{\log_2(1-\epsilon)}{n}.
\]
The diamond norm distance $\epsilon$ can be related to distance in the completely bounded infinity norm (see Supplemental Material for details, which includes Refs. \cite{fuchs1999cryptographic, Sch96, KrW04, paulsen2002completely}), which in turn guarantees $\epsilon$ goes to zero as $n$ approaches infinity. Therefore we obtain the bound $Q(\mathcal N) \le F(R_{\mathcal N}).$ \\

\textit{Conclusions and outlook.--}We have obtained a bound on quantum capacity using fundamental causality considerations. In doing so, we have introduced a new measure of temporal correlations that is analogous to entanglement logarithmic negativity and possesses desired properties that make it useful for studying channel capacities. Studies of spatial correlations have lead to the formulation of many entanglement monotones with different corresponding  applications and operational meanings e.g. distillable entanglement, entanglement cost, squashed entanglement  \cite{horodecki2009quantum, christandl2004squashed}. As a temporal counterpart of quantum correlations, our work initiates research on operational significance of causality measures that might prove useful in a wider range of applications. The causality method applies to arbitrary quantum channels and produces non-trivial upper bounds for any channel. However, in contrast to most other of such bounds, it does not require optimisation. Our result could help to understand the communication rate of complex systems for which optimisation methods are computationally too costly, including quantum networks and quantum communication between many parties \cite{leung2010quantum, hayashi2007quantum, laurenza2017general, pirandola2016capacities, pant2019routing}.

\textit{Acknowledgements--}We would like to thank Andreas Winter, Artur Ekert and Xin Wang for helpful discussions and Mark Wilde and Stefano Pirandola for useful comments on the manuscript. J.F.F. acknowledges support from the Air Force Office of Scientific Research under grant FA2386-15-1-4082. V.V. thanks the Leverhulme Trust, the Oxford Martin School, and Wolfson College, University of Oxford. 
YO also acknowledges support from EPSRC (Grant No. EP/M024261/1) and the QCDA project (EP/R043825/1) which has received funding from the QuantERA ERANET Cofund in Quantum Technologies implemented within the European Union's Horizon 2020 Programme. The authors acknowledge support from Singapore Ministry of Education. This material is based on research funded by the National Research Foundation of Singapore under NRF Award No. NRF-NRFF2013-01. R.P. and V.V. thank EPSRC (UK).\\

\textit{Competing financial interests--} J.F.F. has financial holdings in Horizon Quantum Computing
Pte. Ltd.
\bibliographystyle{unsrt}
\bibliography{caus}

\appendix
\section{Supplemental Material}
\subsection{Causality under encoding and decoding channels}
An important property that we have used in our proof was that the decoding and encoding procedures do not increase causality, so that $F(R_{\mathcal{M}}) \leq F(R_{\mathcal{N}}^{\otimes n})$. To show this, we first establish the following lemma.
\begin{lemma}
	Let $K$ be a linear map from $k$ qubits to $m$ qubits. Then
	\begin{equation}
	( I  \otimes K ) \textrm{SWAP} ^{\otimes k}	(I \otimes K^{\dagger} ) = (K^{\dagger} \otimes I) \textrm{SWAP} ^{\otimes m}(K \otimes I), \label{swap}
	\end{equation}
	where ($A \otimes B$) means that $A$ and $B$ are applied to the first and second subsystems of each of the $\textrm{SWAP}$s respectively.
\end{lemma}

\begin{proof}
	Let $K= \sum_{i=0}^{2^k-1} \sum_{j=0}^{2^m-1} e_{ij} \ket{j}\bra{i}.$ Now the tensor product of $k$ qubit $\textrm{SWAP}$s admits a representation
	\[
	\textrm{SWAP}^{\otimes k} = \sum_{u,v=0}^{2^{k}-1} (\ket{u}\otimes \ket{v})(\bra{v}\otimes \bra{u}).
	\]
	Therefore
	\begin{align}
	&( I ^{\otimes k} \otimes K ) \textrm{SWAP} ^{\otimes k}  ( I ^{\otimes k} \otimes K^{\dagger} ) \nonumber\\
	&=\sum_{i,j,i',j',u,v} ( I \otimes \ket{j} \bra{i})\ket{u} \ket{v} \bra{v}\bra{u}( I \otimes \ket{i'} \bra{j'}) e_{ij}  e_{i'j'}^{*}\nonumber\\
	&=\sum_{j,j'=0}^{2^m-1} \sum_{u,v=0}^{2^k-1} \ket{u} \ket{j} \bra{v} \bra{j'} e_{vj} e_{uj'}^* .\nonumber
	\end{align}
	Similarly evaluating the right hand sign of equation \eqref{swap} we get
	\begin{align}
	&(K^{\dagger} \otimes  I^{\otimes m}) \textrm{SWAP} ^{\otimes m} (K \otimes  I^{\otimes m})\nonumber\\
	&=\sum_{i,j,i',j',u,v} ( \ket{i} \bra{j}\otimes  I)\ket{u} \ket{v} \bra{v}\bra{u}(\ket{j'} \bra{i'}\otimes  I) e_{ij} ^{*} e_{i'j'}\nonumber\\
	&=\sum_{i,i'=0}^{2^k-1} \sum_{u,v=0}^{2^n-1} \ket{i} \ket{v} \bra{i'} \bra{u} e_{i'v} e_{iu}^{*} \nonumber\\
	&=\sum_{j,j'=0}^{2^n-1} \sum_{u,v=0}^{2^k-1} \ket{u} \ket{j} \bra{v} \bra{j'} e_{vj} e_{uj'}^{*},
	\end{align}
	where in the last step we have relabelled the indices. 
\end{proof}
We are now in a position to prove that $F(R_{\mathcal{M}}) \leq F(R_{\mathcal{N}}^{\otimes n})$. 
\begin{lemma}
	Let $\mathcal{E}$ and $\mathcal{D}$ be encoding and decoding channels and $\mathcal{M} = \mathcal{D} \circ \mathcal{N}^{\otimes n} \circ \mathcal{E}$. Then
	\[
	\log_2  \|R_{\mathcal M} \|_1 \leq \log_2 \| R_{\mathcal N} ^{\otimes n} \|_1 .
	\]
\end{lemma}
\begin{proof}
	Consider the trace norm $\| R_{\mathcal M}\|_1$. The decoding procedure is a local operation and therefore from property 4 of $F(R)$, we have
	\[
	\| R_{\mathcal M}\|_1 \leq \| (\mathcal I \otimes (\mathcal N ^{\otimes n} \circ \mathcal{E})) (R_{\mathcal I ^ {\otimes k}})\|_1.
	\]
	Let $\mathcal{E}$ encode $k$ qubits into $m$ qubits. Using Lemma 1
	\begin{align}
	\| (\mathcal I \otimes (\mathcal N ^{\otimes n} \circ \mathcal{E})) (R_{\mathcal I ^ {\otimes k}})\|_1 &=
	\| (\mathcal{E}^{\dagger} \otimes \mathcal N ^{\otimes n} ) (R_{\mathcal I ^ {\otimes m}})\|_1 \nonumber\\
	&= \| (\mathcal{E}^{\dagger} \otimes \mathcal I ) (R_{\mathcal N} ^ {\otimes n})\|_1. \nonumber
	\end{align}
	Decompose $R_{\mathcal N} ^ {\otimes n}$ into its positive and negative part
	\[
	R_{\mathcal N} ^ {\otimes n} = R_+ - R_-,
	\] 
	where both $R_+$ and $R_-$ are positive semi-definite. By the triangle inequality
	\begin{align}
	&\| (\mathcal{E}^{\dagger} \otimes \mathcal I ) (R_{\mathcal N} ^ {\otimes n})\|_1 \leq \| (\mathcal{E}^{\dagger} \otimes \mathcal I  )( R_+) \|_1 +\| (\mathcal{E}^{\dagger} \otimes \mathcal I  ) (R_-) \|_1 \nonumber \\
	&=\tr ((\mathcal{E}^{\dagger} \otimes \mathcal I )( R_+)) + \tr ((\mathcal{E}^{\dagger} \otimes \mathcal I ) (R_-)) \nonumber \\
	&= \tr((\mathcal{E}^{\dagger} \otimes \mathcal I ) (R_+ + R_-)). \nonumber
	\end{align}
	It has been shown in \cite{barnum2000quantum} that in bounding quantum channel capacity, one can restrict $\mathcal{E}$ to be an isometry with only one non-zero Kraus operator, which we denote by $K$. Then the expression $\tr((\mathcal{E}^{\dagger} \otimes \mathcal I ) (R_+ + R_-))$ can be written as
	\begin{align}
	&\tr((K^{\dagger} \otimes  I  ) (R_+ + R_-)  (K \otimes  I  )) \nonumber \\
	&=\tr ((KK^{\dagger} \otimes  I ) (R_+ + R_-)), \nonumber
	\end{align}
	where we used the cyclic property of the trace. Since $P = KK^{\dagger} \otimes \mathcal I ^{\otimes n}$ is a projector,
	\[
	\tr(P (R_+ + R_-))= \tr (P(R_+ + R_-)P) = \|P(R_+ + R_-)P\|_1.
	\]
	Applying H\"older's inequality twice and making use of the fact the infinity norm of a projector equals one, we get
	\begin{align}
	\|P(R_+ + R_-)P\|_1 &\leq \|P\|_{\infty} \|R_+ + R_-\|_1\|P\|_{\infty} \nonumber \\
	&=\|R_+ + R_-\|_1,\nonumber
	\end{align}
	where $\|\cdot\|_{\infty}$ denotes the infinity norm and is equal to the largest singular value of a matrix . Now, $R_+$ and $R_-$ are orthogonal. Hence
	\begin{align}
	\| R_+ + R_-\|_1 = \| R_+ - R_-\|_1 = \| R_{\mathcal N}^{\otimes n}\|_1, \nonumber
	\end{align}
	which leads to
	$
	\| R_{\mathcal M}\|_1 \leq \| R_{\mathcal N}^{\otimes n}\|_1.
	$
	Finally, since logarithm is a monotonic function, the result follows.
\end{proof}
\subsection{Causality bound against the partial transpose bound}
Let us compare our causality bound to the Holevo and Werner bound.
Given a quantum channel $\mathcal N$, and a transpose map $\mathcal T$, 
the Holevo-Werner upper bound on the quantum capacity is 
\begin{align}
Q_{\mathcal T}(\mathcal N) 
=\log_2  \| \mathcal N  \mathcal T \|_{\diamond} 
=\log_2  \| \mathcal I \otimes \mathcal N \mathcal T \|_1 .\nonumber
\end{align}
Using the definition of the induced norm this can be written as
\[
Q_{\mathcal T}(\mathcal N) = \sup_{\rho} \left( \log_2  \| (\mathcal I \otimes \mathcal N \mathcal T) (\rho)\|_1 \right).
\]
Now we can compare this to our bound. In the case of the maximally mixed input the pseudo-density matrix becomes
\[
R_{\mathcal N}= (\mathcal{I} \otimes \mathcal N ) \left(\frac{\textrm{SWAP}^{\otimes k}}{2^k}\right) = (\mathcal{I} \otimes \mathcal N \mathcal T) (\ket{\Phi^{+}}\bra{\Phi^{+}})^{\otimes k},
\]
and therefore the causality bound becomes
\[
F (R_{\mathcal N})  =\log_2  \| (\mathcal{I} \otimes \mathcal N  \mathcal T)  (\ket{\Phi^{+}}\bra{\Phi^{+}})^{\otimes k} \|_1 .
\]
Comparing this to the Holevo and Werner's result it is clear that
$
F (R_{\mathcal N}) \leq Q_{\mathcal T}(\mathcal N),
$
and the two are equal when the supremum is achieved at the $(\ket{\Phi^+} \bra{\Phi^+})^{\otimes k}$.\\

\subsection{Relation between causality bound and the max-Rains bound}
Given a quantum channel $\mathcal N$, the corresponding max-Rains information $R_{\textrm{max}}(\mathcal N)$ is defined in Ref. \cite{wang2019semidefinite} as 
\[
R_{\textrm{max}}(\mathcal N) := \log_2 \Gamma(\mathcal N),
\]
where $\Gamma(\mathcal N)$ is the solution to 
\[
\min \| \tr_B ( V_{SB}+Y_{SB})\|_\infty
\]
subject to $Y_{SB}, V_{SB} \ge 0$, $T_B(V_{SB}-Y_{SB}) \ge J^{\mathcal N}_{SB}$. $T_B$ denotes a partial transpose over subsystem $B$. Here we prove the following relation between the Rains quantity and the causality bound $F$.
\begin{theorem}
	For any quantum channel $\mathcal N$ with Kraus operators $A_k$,
	\[
	R_{\textrm{max}}(\mathcal N) \leq F(R_{\mathcal N^*}),
	\]
	where $\mathcal N^*(X) := \sum_{k}(A_k^*) X^T (A_k^*)^\dagger.$
\end{theorem}
\begin{proof}
	By duality of norms, for any positive semi-definite matrix $M$, we have $\| M\|_{\infty} =  \max \{ \<M,X\> : \| X\|_1 \le 1, X\ge 0 \}$, where $\<M,X\> = \tr M ^\dagger X.$
	Hence it follows that $\Gamma(\mathcal N)$ is the solution to 
	\[
	\min_{Y_{SB},V_{SB} \ge 0 } \max_{G \ge 0 , \|G\| \le 1} \< \tr_B ( V_{SB}+Y_{SB}) , G \>
	\]
	subject to  $T_B(V_{SB}-Y_{SB}) \ge J^{\mathcal N}_{SB}$.
	By the minimax theorem, we know that for any bilinear function $f(x,y)$, 
	$\min_x \max_y f(x,y)= \max_y \min_x f(x,y)$. Therefore, $\Gamma(\mathcal N)$ is the solution to 
	\[
	\max_{G \ge 0 , \|G\| \le 1} \min_{Y_{SB},V_{SB} \ge 0 }  \< \tr_B ( V_{SB}+Y_{SB}) , G \>
	\]
	subject to  $T_B(V_{SB}-Y_{SB}) \ge J^{\mathcal N}_{SB}$.
	Now without loss of generality, 
	$  \< \tr_B ( V_{SB}+Y_{SB}) , G \> =   \<  V_{SB}+Y_{SB} , G \otimes I_B \>$.
	Hence $\Gamma(\mathcal N)$ is the solution to 
	\[
	\max_{G \ge 0 , \|G\| \le 1} \min_{Y_{SB},V_{SB} \ge 0 } \<  V_{SB}+Y_{SB} , G \otimes I_B \>
	\]
	subject to the same condition as above. Since the trace is invariant under transpose of its argument, it is easy to show that Hilbert-Schmidt inner product is invariant under the transpose of both of its arguments. Similarly the Hilbert-Schmidt inner product is also invariant under the partial transpose of both of its arguments. We can prove this by expanding the arguments in any matrix basis, and then applying the linearity of the trace and the multiplicative property of the trace under tensor products. In consequence $\<T_B(C), T_B(D)  \> = \<C,D\>$, and  
	\begin{align}
	&\<  V_{SB}+Y_{SB} , G \otimes I_B \>
	= \< T_B( V_{SB}+Y_{SB}) , T_B(G \otimes I_B) \> \notag\\
	&= \< T_B( V_{SB}-Y_{SB})  +  T_B(2 Y_{SB}) , G \otimes I_B \> . \nonumber
	\end{align}
	Clearly the above is minimized when $Y_{SB} = 0$ and when 
	$T_B( V_{SB}-Y_{SB})  = J^{\mathcal N}_{SB}$. Hence
	$\Gamma(\mathcal N)$ is the solution to 
	\begin{align}
	\max_{G \ge 0 , \|G\| \le 1}   \<  J^{\mathcal N}_{SB} , G \otimes I_B \>
	= 
	\max_{G \ge 0 , \|G\| \le 1}   \< T_B( J^{\mathcal N}_{SB} ) , G \otimes I_B \>,\nonumber
	\end{align} 
	and we find that:
	\[
	\Gamma(\mathcal N) 
	\le 
	\| T_B( J^{\mathcal N}_{SB} )  \|_\infty
	\le 
	\| T_B( J^{\mathcal N}_{SB} )  \|_1.
	\]
	Now let the causality bound for a channel $\mathcal N$ be given by 
	\[
	F(R_{\mathcal N}) = \log_2 C(\mathcal N),
	\]
	where $C(\mathcal N) =  \| (\mathcal I \otimes \mathcal N)({\rm SWAP} ^{\otimes \ell} / 2^\ell)  \|_1$.
	Note that $({\rm SWAP} ^{\otimes \ell} / 2^\ell)$ is just an $\ell$-qubit swap operator. Namely, 
	\begin{align}
	\frac{{\rm SWAP} ^{ \otimes \ell}}{ 2^\ell}
	= \sum_{{\bf x}, {\bf y} \in \{0,1\}^\ell} |{\bf x}\>\<{\bf y}| \otimes |{\bf y}\>\<{\bf x}|,\nonumber
	\end{align}
	where 
	\begin{align}
	|{\bf x}\> = |x_1\> \otimes \dots \otimes |x_\ell\>, \quad
	|{\bf y}\> = |y_1\> \otimes \dots \otimes |y_\ell\>.\nonumber
	\end{align}
	On the other hand, the Choi state of a channel is just the action of $\mathcal I \otimes \mathcal N$ on a maximally entangled state
	\begin{align}
	J^{\mathcal N}_{SB} &= 
	(\mathcal I \otimes \mathcal N)
	\left(\frac{1}{2^{\ell}} \sum_{{\bf x}, {\bf y} \in \{0,1\}^\ell} 
	|{\bf x}\>\<{\bf y}| \otimes |{\bf x}\>\<{\bf y}|
	\right) \notag\\
	& =  
	\frac{1}{2^{\ell}} \sum_{{\bf x}, {\bf y} \in \{0,1\}^\ell} 
	|{\bf x}\>\<{\bf y}| \otimes \mathcal N(|{\bf x}\>\<{\bf y}|).\nonumber
	\end{align}
	Clearly, 
	\begin{align}
	T_B(J^{\mathcal N}_{SB}) & =  
	\frac{1}{2^{\ell}} \sum_{{\bf x}, {\bf y} \in \{0,1\}^\ell} 
	|{\bf x}\>\<{\bf y}| \otimes (\mathcal N(|{\bf x}\>\<{\bf y}|))^T. \nonumber
	\end{align}
	For any argument $X$ of the quantum channel $\mathcal N$ with Kraus operators $A_k$, we have
	\begin{align}
	\mathcal N(X)^T
	= (\sum_{k}A_k X A_k ^\dagger)^T
	&= \sum_{k}A_k^* X^T A_k ^T  \nonumber \\
	&= \sum_{k}(A_k^*) X^T (A_k^*)^\dagger. \nonumber
	\end{align}
	Denote the linear operator $\mathcal N^*$ to be the conjugate channel of $\mathcal N$ where
	\[
	\mathcal N^*(X) = \sum_{k}(A_k^*) X^T (A_k^*)^\dagger.
	\]
	Then it follows that $\mathcal N(X)^T = \mathcal N^*(X^T)$. Hence,
	\begin{align}
	T_B(J^{\mathcal N}_{SB}) & =  
	\frac{1}{2^{\ell}} \sum_{{\bf x}, {\bf y} \in \{0,1\}^\ell} 
	|{\bf x}\>\<{\bf y}| \otimes \mathcal N^*(|{\bf y}\>\<{\bf x}|)
	\notag\\
	& =  
	(\mathcal I \otimes \mathcal N^*)\left(
	\frac{1}{2^{\ell}} \sum_{{\bf x}, {\bf y} \in \{0,1\}^\ell} 
	|{\bf x}\>\<{\bf y}| \otimes \mathcal |{\bf y}\>\<{\bf x}| 
	\right) \notag\\
	& =  
	(\mathcal I \otimes \mathcal N^*)\left(
	\frac{{\rm SWAP}^{\otimes \ell}}{2^{\ell}} \right). \notag
	\end{align}
	and we recover the causality bound
	\[
	\|T_B(J^{\mathcal N}_{SB}) \|_1  = C(\mathcal N^*),
	\]
	It follows that 
	\[
	\Gamma(\mathcal N) \le C(\mathcal N^*).
	\]
\end{proof}
Notice that when $\mathcal N^* = \mathcal N$, then the causality bound is an upper bound on the Rains bound. 
Regardless of whether $\mathcal N^*$ is a quantum channel, the Rains bound for $\mathcal N$ is always upper bounded by the causality bound for the conjugate channel $\mathcal N^*$.

\subsection{Limit of infinite uses of a channel}
Here we show that in the proof of the causality bound, the error parameter $\epsilon$ goes to zero in the limit of large $n$. Now,
\begin{align}
\epsilon&=
\| {\mathcal I} ^{\otimes k} \otimes  (\mathcal{M} -  {\mathcal I}^{\otimes k}) \|_1 \nonumber
\\
&=\sup_{\| X\|_1 = 1}
\|( {\mathcal I} ^{\otimes k} \otimes  (\mathcal{M} -  {\mathcal I}^{\otimes k})) (X) \|_1 \nonumber.
\end{align}
Consider the spectral decomposition of Hermitian operator 
\begin{align}
X = \sum_{i} \lambda_i |\psi_i\>\<\psi_i|\nonumber,
\end{align}
where $\{|\psi_i\>\}_i$ denotes an orthonormal basis, and $\lambda_i$ are the corresponding eigenvalues.
Let $\mathcal A = ( {\mathcal I} ^{\otimes k} \otimes  (\mathcal{M} -  {\mathcal I}^{\otimes k}))$,
then we have
\begin{align}
\epsilon
&= 
\sup_{\{|\psi_i\>\}_i, \sum_{i}|\lambda_i|=1}
\left\|\mathcal A \left( \sum_{i} \lambda_i |\psi_i\>\<\psi_i| \right) \right\|_1 \notag\\
&\le
\sup_{\{|\psi_i\>\}_i, \sum_{i}|\lambda_i|=1} \left( \sum_{i} |\lambda_i|
\|\mathcal A ( |\psi_i\>\<\psi_i| ) \|_1 \right) \notag\\
&\le 
\sup_{ |\psi\> }\|\mathcal A ( |\psi\>\<\psi| ) \|_1 .\nonumber
\end{align}
Since $\mathcal A$ is the difference of two linear maps 
${\mathcal I} ^{\otimes k} \otimes {\mathcal I} ^{\otimes k}$ and 
${\mathcal I} ^{\otimes k} \otimes \mathcal M$, by linearity
we have
\begin{align}
&\sup_{ |\psi\> }\|\mathcal A ( |\psi\>\<\psi| ) \|_1 \nonumber \\
&=
\sup_{ |\psi\> }\|
( {\mathcal I} ^{\otimes k} \otimes {\mathcal I} ^{\otimes k})( |\psi\>\<\psi| )
-
({\mathcal I} ^{\otimes k} \otimes \mathcal M) ( |\psi\>\<\psi| ) \|_1. \nonumber
\end{align}
Within the supremum, we have 1-norm of the difference between two quantum states. 
Recall that there is the inequality that relates the 1-norm of the difference between quantum states to the fidelity between the states. Let
\begin{align}
f(\rho,\sigma) = \tr \sqrt{ \sqrt \rho \sigma \sqrt \rho} \nonumber
\end{align}
denote the fidelity between two positive semidefinite matrices.
If $\rho=|\psi\>\<\psi|$, then $f(\rho,\sigma)= \sqrt{\<\psi|\sigma|\psi\>}.$
Then we have the
Fuchs-van de Graaf inequalities \cite{fuchs1999cryptographic}
\begin{align}
1-f(\rho,\sigma)
\le 
\frac{1}{2}\| \rho - \sigma \|_1
\le
\sqrt{1-f(\rho,\sigma)^2}.\nonumber
\end{align}
Hence
\begin{align} 
&\frac 1 2 \sup_{ |\psi\> }\|\mathcal A ( |\psi\>\<\psi| ) \|_1 \le \nonumber\\
&\sqrt{1- \inf_{ |\psi\> }
	f( ( {\mathcal I} ^{\otimes k} \otimes {\mathcal I} ^{\otimes k})( |\psi\>\<\psi| )  ,
	({\mathcal I} ^{\otimes k} \otimes \mathcal M) ( |\psi\>\<\psi| )   )^2}.\nonumber
\end{align} 
The above inequality is related to entanglement fidelity $F_e(\rho,\Phi)$ of a state $\rho$ with respect to the channel $\Phi$.
Let $\Phi(\rho) = \sum_{A \in K} A \rho A ^\dagger$.
Then from Schumacher's formula \cite{Sch96}, we have
\begin{align}
F_e(\rho,\Phi) &= \<\phi | (\Phi \otimes \mathcal I) (|\phi\>\<\phi| )|\phi \>\nonumber\\
&= f( (\Phi \otimes \mathcal I) (|\phi\>\<\phi| ), |\phi\>\<\phi|) \nonumber\\
&=    \sum_{A \in K} |\tr \rho A|^2,\nonumber
\end{align}
where $|\phi\>$ is a purification of $\rho$.
We denote 
\begin{align}
F_e(\Phi) &=
\inf_\rho F_e(\rho,\Phi)\nonumber\\
&=\inf_{|\phi\>} \<\phi |( \Phi \otimes\mathcal I) (|\phi\>\<\phi| )|\phi \> \nonumber\\
&= 
\inf_{|\phi\>} f( |\phi\>\<\phi|,   (\Phi \otimes \mathcal I) (|\phi\>\<\phi| ) )^2.\nonumber
\end{align}
Hence using the notation for the entanglement fidelity, we have 
\begin{align}
\frac 1 2 \sup_{ |\psi\> }\|\mathcal A ( |\psi\>\<\psi| ) \|_1
\le
\sqrt{1- F_e(  \mathcal M) },\nonumber
\end{align}
thus
$
\epsilon \le 2 \sqrt{1- F_e(  \mathcal M) }.
$
Kretschmann and Werner \cite[Proposition 4.3]{KrW04} showed that
\begin{align}
1- F_e(\Phi) 
\le
4 \sqrt {\|\Phi- \mathcal I \|_{\rm cb}} 
\le 
8 \left( 1- F_e(\Phi)  \right)^{1/4} ,\nonumber
\end{align}
where $\| \cdot \|_{\rm cb}$ denotes the completely bounded norm induced on the operator infinity norm \cite{paulsen2002completely}.
Therefore  
\begin{align}
\epsilon 
\le 
2\sqrt{4 \sqrt {\|\mathcal M- \mathcal I\|_{\rm cb}} }
=
4 \|\mathcal M- \mathcal I\|_{\rm cb}^{1/4}.\nonumber
\end{align}
Since $ \|\mathcal M- \mathcal I\|_{\rm cb}$ is guaranteed to approach zero as $n$ approaches infinity in the channel capacity theorems, $\epsilon$ here also approaches zero.

\end{document}